\documentclass{article}

\usepackage{microtype}
\usepackage{graphicx}
\usepackage{subfigure}
\usepackage{booktabs} 

\usepackage{hyperref}
\usepackage{amsthm}
\usepackage{enumitem,multirow}

\newtheorem{theorem}{Theorem}
\newtheorem{lemma}{Lemma}

\usepackage[accepted]{icml2021}

\usepackage{amsmath, amsthm, amsfonts}

\begin{document}

\twocolumn[
\icmltitle{Learning Self-Modulating Attention in Continuous Time Space with Applications to Sequential Recommendation}

\begin{icmlauthorlist}
\icmlauthor{Chao Chen}{ai} 
\icmlauthor{Haoyu Geng}{ai,cs} 
\icmlauthor{Nianzu Yang}{ai,cs}
\icmlauthor{Junchi Yan}{ai,cs}
\icmlauthor{Daiyue Xue}{mtuan}
\icmlauthor{Jianping Yu}{mtuan}
\icmlauthor{Xiaokang Yang}{ai,cs}
\end{icmlauthorlist}

\icmlaffiliation{ai}{MoE Key Lab of Artificial Intelligence, AI Institute, Shanghai Jiao Tong University, Shanghai, China}
\icmlaffiliation{cs}{Department of Computer Science and Engineering, Shanghai Jiao Tong University, Shanghai, China}
\icmlaffiliation{mtuan}{Meituan, Beijing, China}
\icmlcorrespondingauthor{Junchi Yan}{yanjunchi@sjtu.edu.cn}

\vskip 0.3in
]


\printAffiliationsAndNotice{}  

\begin{abstract}
User interests are usually dynamic in the real world, which poses both theoretical and practical challenges for learning accurate preferences from rich behavior data. Among existing user behavior modeling solutions, attention networks are widely adopted for its effectiveness and relative simplicity. Despite being extensively studied, existing attentions still suffer from two limitations: i) conventional attentions mainly take into account the spatial correlation between user behaviors, regardless the distance between those behaviors in the continuous time space; and ii) these attentions mostly provide a dense and undistinguished distribution over all past behaviors then attentively encode them into the output latent representations. This is however not suitable in practical scenarios where a user's future actions are relevant to a small subset of her/his historical behaviors. In this paper, we propose a novel attention network, named \textit{self-modulating attention}, that models the complex and non-linearly evolving dynamic user preferences. We empirically demonstrate the effectiveness of our method on top-N sequential recommendation tasks, and the results on three large-scale real-world datasets show that our model can achieve state-of-the-art performance.
\end{abstract}

\section{Introduction}\label{sec:intro}
Preference learning is of critical importance in modern machine learning applications, such as online advertising, E-commerce and social media. Often, users' preferences are not static and evolve over time due to a variety of reasons. For example, a user's interest in a particular brand might fade away due to maturity or a change in lifestyle. Nonetheless, it is difficult to model all causal aspects explicitly \cite{widmer1996learning,koren2010collaborative,wu2017recurrent}.

To model the dynamics in user preferences, many research works~\cite{kang2018self,ying2018sequential,zhou2018deep,zhang2019feature} develop different attention mechanisms to learn the sequential patterns of how users' future actions are interacted with previous behaviors, and so far have achieved state-of-the-art performance on many benchmark datasets. The majority of these attentions (implicitly) assume that the underlying sequential patterns are time independent\cite{kumar2019predicting,li2020time,chang2020continuous}. This is however contradicted with the fact that more recent behaviors will have more impact on the future actions, or more specifically identical behavior sequences of two users with different time intervals should have different implications on their current and future decisions.

To tackle the challenging problem, recent approaches generalize the notion of positional embedding to continuous time \cite{xu2019self,shukla2021multi}. The idea behind these approaches is to encode periodic patterns dependent on the progression of time into a high-dimensional vector space (i.e., time embedding), then by using translation-invariant random Fourier features~\cite{rahimi2007random,yu2016orthogonal} to approximately measure the temporal distance between user behaviors in the sequence.

Orthogonal to prior studies, we focus on two pertinent questions fundamental to attention-based preference learning in continuous time space: (i) \textit{how can attentions model the impact of both sequential positions and continuous timestamps in an explicit manner?} --- This matters a lot for business, as it helps to better understand users and their needs~\cite{zhang2020explainable,arrieta2020explainable,chen2021scalable}; and (ii) \textit{how can attentions identify sparse behaviors interacted with future actions?} --- Existing techniques in this direction mainly base on sparsemax \cite{martins2016softmax,niculae2017regularized} to yield sparse distributions over past behaviors. These approaches tend to preserve very limited temporal information which leads to the difficulty to capture complex and fine-grained temporal dynamics.

In this paper, we propose a novel attention mechanism, named self-modulating attention, to model the complex evolution of user preferences over time. More specifically, our attention explicitly digests temporal dynamics via a conditional intensity function that specifies whether the observed behaviors are relevant with users' future interests at a given time point: high (low) intensities of the behaviors amplify (attenuate) their contributions in the output representations. Here the term \textit{self-modulating} is meant to encompass above positive and negative effects of a single sequence to compute the representation of the same sequence. We further propose continuous time regularization to penalize intensities, as the source of supervised signal for preference learning is mostly coming from the behavior data that is independent of time.

\textbf{Intuition.}
Existing attention mechanism is based on the correlation between the query and keys (or between the past and current events for sequential behavior modeling), while such correlation is not modeled as a function of time. However, for real-world behavior modeling e.g. sequential recommendation in the continuous time space, a user's behavior pattern may vary over time.

\textbf{Theoretical Motivation.}
More interestingly, in this paper we theoretically show (see details in Theorem \ref{thm:bound}) under a bounded impact of the past behavior to future actions, more sparse behaviors can lead to higher generalization ability for a prediction model. This theory motivates us to design a mechanism for more effectively uncovering such truly useful past behaviors in a more discriminative manner. 

\textbf{Technical Implementation.}
Based on the above intuitive and theoretical analysis, in this paper we aim to develop a time sensitive attention mechanism to adaptively and predictively re-weight the past behaviors for their impact to current scoring. The temporal point process is used for its interpretable and rigorous Bayesian nature. The conditional intensity function is used to modulate and soften the time-varying attention in continuous time space.

The key contributions of this paper are as follows:
\begin{itemize}[topsep=0in,leftmargin=0em,wide=0em,parsep=0em]
    \item We establish a theoretical bound for attention based approaches in preference learning, showing the benefits of bounded and sparsified impact of behavior history. 
    
    \item To better explore the behavior history, we develop a notion of self-modulating attention, with the classic attention as a special case of our model. Our work concerns modeling attentions informed by time-dependent uncertainties.
    
    \item To effectively incorporate the temporal model into preference learning, we develop a continuous time regularizer to fill the gap between traditional preference learning and continuous-time temporal process modeling. 
    
    \item Top-N recommendation results on benchmark datasets demonstrate that our self-modulating attention outperforms state-of-the-art methods by a notable margin, and it can also be reused as an orthogonal plug-in for existing methods. 
\end{itemize}

\section{Related Work}
User modeling is of crucial importance in many online applications such as advertising, E-commerce and social media. In general, the goal of user modeling is to learn user representations from the complex and abundant behavior data \cite{Adomavicius05,Su09}. 

{\bf Traditional Recommendation Approaches.}
Traditional approaches usually assume that user preference is static, and collaborative filtering (CF)~\cite{Herlocker99,sarwar2001item} is one of the most popular approaches due to ease of implementation and quality of recommendation. Among existing CF solutions, matrix approximation based methods~\cite{koren2008factorization,koren2009matrix,chen2015wemarec,li2017mixture,li2019mixture} have achieved state-of-the-art performance in many benchmark datasets, especial for Netflix prize data~\cite{Netflix07}. However, in practice, user preferences often drift over time due to various reasons. 

{\bf Sequential Recommendation Approaches.}
To address the problem, FPMC~\cite{HeICDM16} and its hierarchical version HRM~\cite{WangSIGIR15}, propose to use Markov chains to model sequential patterns by learning user-specific transition matrices. These approaches have been proved to be able to capture short-term patterns. The major concern behind them is the potential state space explosion intractability in face of different possible sequences over items \cite{he2017translation,wu2017recurrent,chen2021dyna}.

Another emerging line focuses on modeling sequential patterns by using recurrent neural networks. For example, RRN \cite{wu2017recurrent} and GRU4REC \cite{HidasiICLR16} exploit LSTM~\cite{hochreiter1997long} or GRU \cite{chung2014empirical} module to capture dynamic user preferences dependent on sequence positions. Meanwhile, SHAN \cite{ying2018sequential}, RUM~\cite{chen2018sequential}, DIN~\cite{zhou2018deep} and SASREC~\cite{kang2018self} introduce attention mechanisms to sequential recommendation problems. Compared with recurrent networks, attentions have shorter path to access distant positions, and thereby achieve superior performance in many cases~\cite{vaswani2017attention}.

{\bf Temporal Recommendation Approaches.}
Compared to sequential recommendation problem that has been well studied, relatively less attention has been paid to temporal recommendation algorithms. The prominent work timeSVD++ \cite{koren2010collaborative} proposes the use of explicit temporal bias to model the discrete time dynamics. Apart from hand engineered features, the authors in \cite{lu2016collaborative} develop a vector auto-regressive model to predict future user preferences. Empirical results have shown state-of-the-art performance on many benchmark datasets, but the costly computational overhead and memory use make it hard to support large-scale industrial applications. 

Another line of works follow the idea of positional embedding which designs time embedding to capture temporal information, with attention mechanisms for end-to-end learning~\cite{liu2020kalman}. The authors in \cite{li2020time} divide the time into intervals and project them to low-dimensional space. This treatment discretizes the time and reduces the power of modeling fine-grained temporal information~\cite{trivedi2019dyrep,chang2020continuous}. By contrast, the works
\cite{xu2019self,shukla2021multi} learn functions of continuous time to capture the dynamics.

Distinct from the above approaches, our method involves re-weighting of attention coefficients according to the intensity function of temporal point processes~\cite{daley2003introduction}, such that self-modulation can be achieved based on continuous time information. Thanks to the intensity powered by neural networks, our self-modulating attention is enabled to model complex evolving dynamics in continuous time space. 

Temporal point process has been shown powerful for temporal modeling. The work \cite{zhou2013learning} deals with
mutual-excitation of event sequences, with a regularizer on the sparseness and low-rank assumption of mutual correlations. 
As the rise of deep learning, recurrent neural networks are widely adopted to model temporal dynamics~\cite{mei2017neural,trivedi2017know,xiao2017modeling}, which are further replaced or enhanced by attentions~\cite{zhang2020self,zuo2020transformer}. We argue that these works are in general orthogonal to ours, as they focus on the modeling of event times rather than preference learning. In fact, our work can enjoy the advances in \cite{zhang2020self,zuo2020transformer} to improve the expressiveness. 

This paper focuses on developing a new attention dependent on both sequential positions and continuous timestamps. More specifically, we generalize the formulation of existing attentions by introducing the intensity function, which specifies if user tastes on the items are changed or not. Compared to existing sequential recommendation models, the difference lies in the ability to model temporal dynamics in the continuous time space. 

\section{Notations and Preliminaries}
\label{sec:prelim}
Throughout this paper, we denote scalars by either lowercase or uppercase letters, vectors by boldface lowercase letters, and matrices by boldface uppercase letters. Unless otherwise specified, all the vectors are column vectors. 

Let $\mathcal{S}\!=\!\{(t_j, i_j)\}_{j=1}^L$ be a user behavior sequence, each pair $(t_j, i_j)$ denotes that item $i_j$ is rated at time $t_j$.  Then the attentions can be generally defined by, given any query vector $\mathbf{q}$ and the sequence $\mathcal{S}$:
\begin{flalign}
\mathrm{Att}(\mathbf{q}\big|\mathcal{S}) 
= \sum_{j'=1}^n
    p(\mathbf{v}_{j'}\big|\mathcal{S}) \mathbf{v}_{j'}
= \mathbb{E}_{p(\mathbf{v}|\mathcal{S})} \mathbf{v},
\end{flalign}where $n$ is the number of items and the distribution $p(\mathbf{v}|\mathcal{S})$ attentively aggregates all the value vectors $\{\mathbf{v}_{j'}\}_{j'=1}^n$. 

In the scaled dot-product attention \cite{vaswani2017attention}, $p(\mathbf{v_{j'}}|\mathcal{S})$ equals to $\eta\exp(\mathbf{q}^\top\mathbf{k}_{j'}/\sqrt{d})$ if the item $j'$ is in the sequence $\mathcal{S}$ and equals to zero otherwise, where $\eta$ is used for normalization i.e., $\eta^{-1} = \sum_{(t,i)\in\mathcal{S}} \exp(\mathbf{q}^\top\mathbf{k}_{i}/\sqrt{d})$. In the masked version, we should use $p(\mathbf{v_{j'}}|\mathcal{H}_t)$ where $\mathcal{H}_t = \{(t_j, i_j) \big| t_j < t\}$ signifies the history up to time $t$.

In the vast literature of sequential recommendation, attention mechanisms have been widely adopted to capture the dynamic user preferences. Most of these approaches \cite{ying2018sequential,kang2018self,ma2019hierarchical} can be generally formulated as follows:
\begin{flalign}\label{eq:srec}
\hat{\mathbf{R}} = \mathbf{PVB}
\quad\mathrm{and}\quad
\mathbf{P}_{u,i} = p(\mathbf{v}_i\big|\mathcal{H}_t^{(u)}),
\end{flalign}where $\hat{\mathbf{R}}$ denotes the estimate of the underlying matrix $\mathbf{R}$ with $m$ users and $n$ items, $\mathbf{P}\in \mathbb{R}^{m\times n}$ signifies the correlations between the past behaviors and future actions, $\mathbf{V}\in\mathbb{R}^{n\times{d}}$ denotes the value matrix, and $\mathbf{B}\in\mathbb{R}^{d\times{n}}$ acts as the output transformation matrix and can also be viewed as item feature matrix. Notably, the nonlinear projection $\phi(\mathbf{PV}, \mathbf{B})$ as in \cite{he2017neural}, can also be formed as Eq.~(\ref{eq:srec}). This is because the ranking score is usually a non-negative measure, so that it can be approximated by the inner product of random Fourier feature matrices~\cite{rahimi2007random,yu2016orthogonal,choromanski2020rethinking}.

{\bf Remark.} We continue to study the generalization bound of attention-based models in the form of Eq.~(\ref{eq:srec}). Mathematically, we set $\rho$ to bound the number of non-zero entries in $\mathbf{P}$ that characterizes sequential patterns, and $\mu$ to bound the impact of each historical behavior. As proved in Theorem \ref{thm:bound}, the generalization error of these models decreases with both $\rho$ and $\mu$. This result is consistent with \cite{lee2013local,li2016low} --- penalizing the model parameters with $\ell_2$-norm can usually improve the model performance by avoiding overfitting. Perhaps more importantly, the theorem shows that it is beneficial to limit a user's future action correlated with a small subset of past behaviors.

\begin{theorem}[\textbf{Generalization Bound}]\label{thm:bound}Suppose that the loss function $\ell$ is $L$-Lipschitz, and for the estimate $\hat{\mathbf{R}}$ on an random example set $\Omega$ we bound $\rho|\Omega|=\sup_{\mathbf{P}\in\mathcal{F}} \sum_{(u,i)\in\Omega}\parallel\mathbf{P}_{i, *}\parallel_0$ and $\mu=\sup_{(u,i,k)\in\Omega} |\mathbf{P}_{u,k} (\mathbf{VB})_{k,j}|$, then with probability at least $1 - \delta$, we have the bound:
\begin{equation}
\begin{aligned}
    \mathbf{E} \Big[
        \ell(\mathbf{R}, \hat{\mathbf{R}}) \Big]
    \le \mathbf{E}_\Omega &\Big[
        \ell(\mathbf{R}, \hat{\mathbf{R}}) \Big] 
    + \mathcal{O} \Big(
        L\mu\sqrt{\frac{C\rho\ln|\Omega|}{|\Omega|}} \\
    &+ \sqrt{\frac{\ln(1/\delta)}{|\Omega|}}
    \Big)
\end{aligned}
\end{equation}
where $C = d(m+n)\log(48emn)$.
\end{theorem}

\section{Methodology}\label{sec:method}
In this section, we first formulate the generalized attention in continuous time space, which has been rarely studied before. Then we introduce a specific implementation, called self-modulating layer (SMLayer), followed by a continuous time regularization for preference learning, making it more tailored to the sequential recommendation problem.

\subsection{Notations}
Recall that $\mathcal{S}=\{(t_j, i_j)\}_{j=1}^L$ denote a user's behavior sequence of size $L$, in which each pair $(t_j, i_j)$ means that the given user clicks/likes/views item $i_j$ at time $t_j$. Then, we further define $\mathcal{H}_t = \{(t_j, i_j) | t_j < t\}$ that signifies the history up to time $t$. The goal of self-modulating attention is to model the impact of sequential positions as well as continuous timestamps. 

We recall that user preferences evolve over time and the interests on the items might wane, such that some past behaviors should not make contributions to the output user representations at some time, even though they are semantically correlated. To remedy such challenge, one straightforward idea is to directly inform the correlation $p(\mathbf{v}_{i_j}\big|\mathcal{H}_{t})$ with the expected number of occurrence $\mathbb{E}[\mathrm{N}(t, t + \mathrm{d}t)\big| \mathcal{H}_{t-},\mathbf{v}_{i_j}]$ conditional on the history $\mathcal{H}_{t-}$, where $\mathcal{H}_{t-} = \mathcal{H}_{t} \cup \{t_{j+1}\notin(t_j, t)\}$ and $\mathrm{N}(t, t + \mathrm{d}t) \in \{0, 1\}$ denotes the number of occurrences for item $i_j$ in an infinitesimal interval.

\subsection{Generalizing Attention to Continuous Time Space}
Following this idea, we make an adjustment for classic attentions~\cite{chorowski2015attention,vaswani2017attention,xu2019self} to capture both sequential and temporal dynamics, which is devised as follows:
\begin{flalign}
&\mathbb{E}_{p(\mathbf{v}\big|\mathcal{H}_{t})}
    \Bigg[ 
        \mathbb{E}\Big[\mathrm{N}(t, t + \mathrm{d}t)\big| \mathcal{H}_{t-},\mathbf{v}\Big]\mathbf{v}
    \Bigg] \nonumber\\
\stackrel{(a)}{=}& \mathbb{E}_{p(\mathbf{v}\big|\mathcal{H}_{t})}
    \Bigg[ 
        p \Big( t_{j+1}\in[t, t + \mathrm{d}t]\big| \mathcal{H}_{t-}, \mathbf{v}\Big)\mathbf{v}
    \Bigg] \nonumber\\
\stackrel{(b)}{=}& \mathbb{E}_{p(\mathbf{v}\big|\mathcal{H}_{t})}
        \frac{
            p \Big( t_{j+1}\in[t, t + \mathrm{d}t]\big|\mathcal{H}_{t},\mathbf{v}\Big)
        }{
            p \Big(t_{j+1}\notin(t_j, t) \big| \mathcal{H}_{t},\mathbf{v}\Big)
        }\mathbf{v} \nonumber\\
\stackrel{(c)}{=}& \mathbb{E}_{p(\mathbf{v}\big|\mathcal{H}_{t})}
        \frac{
            f(t\big|\mathcal{H}_{t},\mathbf{v})\mathrm{d}t
        }{
            S(t\big|\mathcal{H}_{t}, \mathbf{v})
        }\mathbf{v} \nonumber\\
\stackrel{(d)}{=}  &
\mathbb{E}_{p(\mathbf{v}\big|\mathcal{H}_{t})}     \mathbf{v} 
    \lambda^\ast(t\big|\mathcal{H}_{t},\mathbf{v})\mathrm{d}t, \label{eqn:sm_defn}
\end{flalign}
in which (a) holds due to the fact that $\mathrm{N}(t_j, t_j + \mathrm{d}t)$ is either zero or one in the infinitesimal interval; (b) holds due to Bayes' theorem; and (c) - (d) hold due to 
\begin{equation}\begin{aligned}
p \Big( t_{j+1}\in[t, t + \mathrm{d}t]\big|\mathcal{H}_{t},\mathbf{v}\Big) 
    &= f(t\big|\mathcal{H}_{t},\mathbf{v})\mathrm{d}t \\
p \Big(t_{j+1}\notin(t_j, t)\big| \mathcal{H}_{t},\mathbf{v}\Big) 
    &= S(t\big|\mathcal{H}_{t}, \mathbf{v}) \\
\frac{
            f(t\big|\mathcal{H}_{t},\mathbf{v})
        }{
            S(t\big|\mathcal{H}_{t}, \mathbf{v})
        }.
    &= \lambda^\ast(t\big|\mathcal{H}_{t},\mathbf{v}),
\end{aligned}\end{equation}
where $S(t\big|\mathcal{H}_{t}, \mathbf{v})$ and $\lambda^\ast(t\big|\mathcal{H}_{t},\mathbf{v})$ are called survival function and conditional intensity function respectively, in the literature of temporal point process \cite{aalen2008survival,kleinbaum2010survival}.

In general, dynamic processes will be characterized by the conditional intensity function $\lambda^*(t\big|\mathcal{H}_{t})$. Within a short time window $[t,t+dt)$, $\lambda^*(t)$ represents the occurrence rate of a new event given history $\mathcal{H}_t$ that 
\begin{equation}
\lambda^*(t\big|\mathcal{H}_{t})
=\frac{p(\mathrm{N}(t+\mathrm{d}t)-\mathrm{N}(t)=1|\mathcal{H}_t)}{\mathrm{d}t},
\end{equation}
where $\mathrm{N}(t)$ is the counting process and $*$ reminds it is history dependent. Here we omit $\mathbf{v}$ for its generality.

Using the generalized formulation of conventional attentions, we can deal with the uncertainty over time and are able to predict future user preferences at time $t$ by
\begin{gather}\label{eqn:smatt_discr}
\widetilde{\mathrm{Att}}(\mathbf{q}, \big| \mathcal{H}_{t})
    = \mathbb{E}_{p(\mathbf{v}\big|\mathcal{H}_{t})} \mathbf{v}
    \lambda^\ast(t\big|\mathcal{H}_{t},\mathbf{v}),
\end{gather}
or future user preferences during time interval $(t_j, t]$ by
\begin{gather}\label{eqn:smatt_contiu}
\widetilde{\mathrm{Att}}(\mathbf{q}, \big| \mathcal{H}_{t})
    = \mathbb{E}_{p(\mathbf{v}\big|\mathcal{H}_{t})} \mathbf{v}
    \int_{t_j}^t\lambda^\ast(s\big|\mathcal{H}_{t},\mathbf{v})\,\mathrm{d}s,
\end{gather}where the integral is challenging to compute and proper approximations are needed which will be discussed later.

{\bf Remark.} We here discuss the characteristics of our proposed self-modulating attention: (i) the conditional intensity function $\lambda^\ast(t\big|\mathcal{H}_{t},\mathbf{v})$ models temporal dynamics in an explicit manner, which acts as a role of time-dependent modulating factor, amplifying or attenuating the impact of past behaviors. This helps to derive the explanations behind the model predictions and the results are highlighted in the experiment sections; (ii) The intensities equal to zero can help to detect the sparse structures in the correlations between past behaviors and future actions, while the intensities less than one, can help to reduce the contribution of the behaviors when computing the predictions. As shown in Theorem \ref{thm:bound}, these two properties provide improved model generalization ability. 

\subsection{Self-modulating Layer}
The conditional intensity of a temporal point process is often designed to capture the phenomena of interest, for example Poisson process, Hawkes process~\cite{hawkes1971spectra,mei2017neural,zhou2013learning} and Self-correction process~\cite{isham1979self}. In our case, we implement the intensity function as a neural layer, named self-modulating layer (SMLayer), which can be stacked on the attentions. By doing so, intermediate representations of the attentions can be reused for computational efficiency, and moreover this makes SMLayer a general building block for the broader applications.

We use $\mathbf{Y}\in\mathbb{R}^{L\times{d}}$ to denote the embeddings of the input sequence $\mathcal{S}=\{(t_j, i_j)\}_{j=1}^L$. As our model still needs to be aware of positional information, we choose to use the positional encoding as adopted in~\cite{vaswani2017attention}:
\begin{align}
\mathbf{Z}_{i,j} = \left\{
    \begin{array}{ll}
        \sin(i/10000^{j/d}), 
        \quad\quad\,\mathrm{if}~i~ \mathrm{is~even},\\
        \cos(i/10000^{j-1/d}), 
        \quad\mathrm{if}~i~ \mathrm{is~odd}.
    \end{array}
\right.
\end{align}After the initial embedding and positional encoding layer, we will have the input embedding $\mathbf{X}$, which is then passed to the self-attention module to compute the intermediate representations
\begin{gather}
\mathbf{X} = \mathrm{concat}\big(\,[\mathbf{Y}, \mathbf{Z}]\,\big) \nonumber\\
\mathbf{Q} = \mathbf{XW}^{Q},\quad\mathbf{K} = \mathbf{XW}^{K},\quad\mathbf{V}^\mathrm{seq} = \mathbf{XW}^{V} \nonumber\\
\mathbf{H} = \mathrm{softmax}(\frac{\mathbf{QK^\top}}{\sqrt{d}})\mathbf{V}^\mathrm{seq}.\nonumber
\end{gather}where all $\mathbf{W}$ are projection matrices, $\mathbf{Q},\mathbf{K}, \mathbf{V}^\mathrm{seq}$ are separately the query, key and values matrices obtained by different transformations of the input $\mathbf{X}$, and $\mathbf{H}$ is the output representations of conventional attentions, where each row corresponds to a particular user behavior. Note that to ensure causality, we mask the future information.

Before the next step, we first denote the $j^\mathrm{th}$ row of the self-attention output $\mathbf{H}$ by $\mathbf{h}(t_j)$, which corresponds to the representation of the sub-sequence $\{(i_{j'}, t_{j'})\}_{j'=1}^j$ up to the time $t_j$. To model the continuous time temporal dynamics, the intermediate representation $\mathbf{h}(t_j)$ is fed through item-wise feed-forward neural network
\begin{flalign} \label{eqn:att_g}
\mathbf{g}_{k}(t) = \sigma\left(
        \underbrace{\mathbf{W}^G_{k}\mathbf{h}(t_j)}_\mathrm{Endogenous}
        + \underbrace{\mathbf{b}^G_{k} (t - {t}_{j})}_\mathrm{Exogenous}
    \right),
\end{flalign}where the endogenous correlation with the potential next item $k$ is captured by the term $\mathbf{W}^G_k\mathbf{h}(t_j)$, which evolves in the embedding space with respect to its previous behaviors in the sequence and not in a random fashion. In the meantime, $\mathbf{b}^G_{k} (t - {t}_{j})$ are used to capture the exogenous force that reacts on the potential next item $k$ during the time interval $[t_j, t)$. By combining $\mathbf{W}^G_k\mathbf{h}(t_j)$ and $\mathbf{b}^G_{k} (t - {t}_{j})$ together, the output representation $\mathbf{g}_{k}(t)$ is expressive to better characterize both of the sequential and temporal dynamics in the continuous time space.

Lastly, the conditional intensity of each potential next item $k$ can be computed in the following function
\begin{flalign}\label{eqn:att_l}
\mathbf{\lambda}^\ast(t\big|\mathcal{H}_{t}, \mathbf{v}_k)
= f_{k} (\mathbf{w}^\top_{k}\mathbf{g}_k(t) + \mu_k ),
\end{flalign}where $\mathbf{w}_k$ is the model parameter that serves for the item $k$, $\mu_k$ signifies the base rate that the item $k$ appears in the future, and $f_k$ is the activation function. We notice that the choice of activation function should consider two critical criteria: i) the intensity should be non-negative and ii) the dynamics on different items evolve at different scales. To account for this, we adopt the softplus function
\begin{gather}
    f_k(x) = \phi_k \log\left(
        {1} + \exp(x/\phi_k)
    \right),
\end{gather}where the parameter $\phi_k$ captures the timescale difference.

{\bf Remark.} In practice because the number of items is always very large, Eq. (\ref{eqn:att_g} - \ref{eqn:att_l}) require considerable computational time and memory use, not suitable for large-scale online systems. One of the feasible solutions is to cluster the items into groups using spectral clustering techniques \cite{shi2000normalized,chen2010parallel}, and replace the item-wise feed-forward neural networks with group-wise feed-forward neural networks. This makes sense because the items in the same group are densely connected to each other but sparsely connected to the items from different groups.

\subsection{Continuous Time Regularization for Preference Learning}
In preference learning, the model parameters are usually optimized by minimizing the reconstruction error~\cite{mackey2011divide,lee2013local,li2016low,chen2016mpma}. The only source of supervised signal is from the behavior data (i.e., $\mathbf{R}$) that is independent of time. Hence, the intensity function learned in such protocol might probably diverge from the complex continuous-time patterns contained in the data.

To remedy the problem, we propose the continuous time regularization $R(\Theta)$ which maximizes the log-likelihood of the timestamps $\{t_j\}_{j=1}^L$ in the given sequence, defined by
\begin{gather}
{R}(\Theta;u) 
    =\sum_{j=1}^L \log \lambda^\ast_{k_{j}}(t_j\big| \mathcal{H}^{(u)}_{j})
    - \int_{t_1}^{t_L}\lambda^\ast(t_j\big| \mathcal{H}^{(u)}_{j}) \mathrm{d}t \nonumber
\end{gather}where $\mathcal{H}^{(u)}_j\!=\!\{(t_i, i_i)\big|t_i\!<\!t_j\}$ is the subsequence up to time $t_j$. The first term in the right hand signifies the log-likelihood of the behavior history, where $k_j$ signifies the identity of item $i_j$ or the group it belongs to. The second term represents the log-survival probabilities where $\lambda^\ast(t_j\big| \mathcal{H}^{(u)}_{t_j})
= \sum_k \lambda^\ast_{k}(t_j\big| \mathcal{H}^{(u)}_{j})$. Concretely, if we are given the sequences of $m$ users, the model parameters can be optimized by minimizing the following objective function
\begin{equation}\label{eqn:ctr_loss}
    \min_\Theta \ell(\mathbf{R}, \hat{\mathbf{R}}) 
    - \gamma \mathbb{E}_{u\in[1,m]}\,R(\Theta; u),
\end{equation}where $\gamma$ is the regularization parameter.

The challenge in optimizing Eq.~(\ref{eqn:ctr_loss}) is to compute the integral $\Lambda=\int_{t_1}^{t_L}\lambda^\ast(t_j\big| \mathcal{H}^{(u)}_{j}) \mathrm{d}t$. Due to the softplus used in Eq.~(\ref{eqn:att_l}), there is no closed-form solution for this integral.

We consider two techniques to approximate the integral $\Lambda$.

\paragraph{1) Monte Carlo integration}
\cite{metropolis1949monte,robert2013monte}. 
It randomly draws a set of samples $\{v_i\}_{i=1}^N$ in each interval $(t_{j-1}, t_j)$ to yield an unbiased estimation of $\Lambda$, i.e., $\mathbb{E}\,\widetilde{\Lambda}_\mathrm{MC} = \Lambda$
\begin{gather}\label{eqn:mc}
   \widetilde{\Lambda}_\mathrm{MC} 
        = \sum_{j=2}^L 
            (t_j - t_{j-1})\Big(
            \frac{1}{N} \sum_{i=1}^{N} \lambda^\ast(v_i
                \big|\mathcal{H}^{(u)}_j)
            \Big);
\end{gather}

\paragraph{2) Numerical integration method}
\cite{stoer2013introduction}.
It is usually biased but fast due to the elimination of sampling. For example, trapezoidal rule approximates the integral using linear functions
\begin{gather}\label{eqn:nu}
   \widetilde{\Lambda}_\mathrm{NU}
=\sum_{j=2}^L
    \frac{t_j-t_{j-1}}{2}
    \Big(
        \lambda^\ast(t_j\big|\mathcal{H}^{(u)}_j)
        + \lambda^\ast(t_{j-1}\big|\mathcal{H}^{(u)}_{j-1})
    \Big).
\end{gather}

In our experiments, we find that the approximation defined in Eq.~(\ref{eqn:nu}) performs comparable to Eq.~(\ref{eqn:mc}) while consuming less training time. Hence we choose the numerical integration in our main approach. We conjecture the reason is that the softplus function is highly smoothed and the bias introduced by linear approximations is relatively small.

\begin{table}[t!]
    \caption{\label{tbl:stats}Statistics of the public datasets for evaluation.}
    \begin{center}
    \begin{tabular}{lccc}\toprule
         Dataset & \#Users & \#Items & \#Interactions \\\midrule
         \textbf{Amazon} & 211,384 
                         & 18,490 & 1.6M\\
         \textbf{Koubei} & 212,831 
                         & 10,213 & 1.8M \\
         \textbf{Tmall}  & 320,497 
                         & 21,876 & 7,6M\\\bottomrule
    \end{tabular}\end{center}
\end{table}
\section{Experiments}
\label{sec:expr}
\begin{table*}[ht!]
    \caption{\label{tbl:abl}Ablation study on the Amazon, Koubei and Tmall datasets. The proposed self-modulating layer (SMLayer) and continuous-time regularization (CTReg) are adapted to attention-based DIN~\cite{zhou2018deep} and SASREC~\cite{kang2018self} models. The performance is evaluated in terms of HR@10 and NDCG@10.}
    \begin{center}\begin{tabular}{cl ccc c ccc} \toprule
         && 
         \multicolumn{3}{c}{\textbf{DIN}}  && 
         \multicolumn{3}{c}{\textbf{SASREC}} \\\cmidrule{3-5}\cmidrule{7-9}
         
         & Dataset &
         Origin & +SMLayer & +CTReg &&
         Origin & +SMLayer & +CTReg \\\midrule
         
         \multirow{3}{*}{\centering\rotatebox{90}{\textbf{HR}}} & \textbf{Amazon} &  
         0.21955 & \textbf{0.22065} & 0.21985 && 
         0.25595 & \textbf{0.26545} & 0.26058 \\
         &\textbf{Koubei} &  
         0.32665 & \textbf{0.33940} & 0.33780 && 
         0.35455 & 0.36194 & \textbf{0.36235} \\
         &\textbf{Tmall}  &
         0.48460 & 0.49033 & \textbf{0.49157}&& 
         0.50433 & 0.51218 & \textbf{0.51347}
         \\\midrule
         
         \multirow{3}{*}{\centering\rotatebox{90}{\textbf{NDCG}}} & \textbf{Amazon} &   
         0.13443 & \textbf{0.13383} & 0.13296 && 
         0.16131 & 0.16475 & \textbf{0.16529} \\
         & \textbf{Koubei} &  
         0.24186 & \textbf{0.25444} & 0.25411 && 
         0.27070 & 0.27862 & \textbf{0.28083} \\
         & \textbf{Tmall}  & 
         0.33855 & 0.34580 & \textbf{0.35062} && 
         0.34326 & 0.35214 & \textbf{0.35811}\\\bottomrule
    \end{tabular}\end{center}
\end{table*}
\begin{table*}[ht!]
    \vskip 0.15in
    \caption{\label{tab:cmp_ac}Performance comparison between the baselines and our proposed method on the Amazon, Koubei and Tmall datasets in terms of HR@10 and NDCG@10. Boldfaces mean that the method performs statistically significantly better under t-tests, at the level of 95\% confidence level. We emphasize the comparison against SASREC+, a variant of SASREC equipped with functional time embedding~\cite{xu2019self} which captures continuous-time temporal dynamics. }
    \begin{center}\begin{tabular}{l cc c cc c cc} \toprule
         & 
         \multicolumn{2}{c}{\textbf{Amazon}}  && 
         \multicolumn{2}{c}{\textbf{Koubei}} &&
         \multicolumn{2}{c}{\textbf{Tmall}} \\\cmidrule{2-3}\cmidrule{5-6}\cmidrule{8-9}
         Model   &
         HR & NDCG &&
         HR & NDCG &&
         HR & NDCG \\\midrule
         \textbf{SHAN}~\cite{ying2018sequential}  & 
         0.19250 & 0.11724 && 
         0.28150 & 0.20256 && 
         0.37316 & 0.25840 \\
         \textbf{DIN}~\cite{zhou2018deep} &  
         0.21955 & 0.13443 && 
         0.32665 & 0.24186 && 
         0.48460 & 0.33855  \\
         \textbf{GRU4REC}~\cite{HidasiICLR16} &  
         0.24380 & 0.15822 && 
         0.32655 & 0.27052 && 
         0.46877 & 0.33746 \\\midrule
         \textbf{SASREC}~\cite{kang2018self}  & 
         0.25595 & 0.16131 && 
         0.35455 & 0.27070 && 
         0.50433 & 0.34326 \\
         \textbf{SASREC+}~\cite{xu2019self}  & 
         0.25820 & 0.16204 && 
         0.35690 & 0.27148 && 
         0.50607 & 0.34328 \\
         \textbf{SASREC w/ ours.}  &
         \textbf{0.26545} & \textbf{0.16529} && 
         \textbf{0.36235} & \textbf{0.28083} && 
         \textbf{0.51347} & \textbf{0.35811} \\\bottomrule
    \end{tabular}\end{center}
\end{table*}
This section evaluates the quantitative and qualitative results of our self-modulating attention against many state-of-the-art baselines on three large-scale real-world datasets. All experiments are run on a machine with E5-2678 CPU, RTX 2080 and 188G RAM. 

\subsection{Benchmark Datasets}
We use three real-world datasets which are processed in line with  \cite{weimer2007cofirank,liang2018variational,steck2019markov}: 
(1) Amazon Electronics\footnote{http://jmcauley.ucsd.edu/data/amazon/} is a user review dataset, where we binarize the explicit data and only keep users and items who have at least 5 historical records; 
(2) Koubei\footnote{https://tianchi.aliyun.com/dataset/dataDetail?dataId=53} is a user behavior dataset, where we keep users with at least 5 records and items that have been purchased by at least 100 users;  
and (3) Tmall\footnote{https://tianchi.aliyun.com/dataset/dataDetail?dataId=35680} is a user click dataset, where we keep users who click at least 10 items and items which have been seen by at least 200 users. Table~\ref{tbl:stats} shows the statistics.

\subsection{Evaluation Protocol}
We study the performance using strong generalization protocol \cite{weimer2007cofirank,liang2018variational,steck2019markov}, where the models are evaluated on the users that are not present at training time. To do so, we split the users into training/validation/test sets with the ratio $8:1:1$. Then, we use all the data from the training users to optimize the model parameters, whereas for validation/test users, we keep the last record of each user to test and the rest to learn necessary representations for the model. In addition, we select model hyper-parameters and architectures using the validation users, and we report the results on test users for the model that achieves the best results on the validation users.

Two ranking-based metrics are used: Hit Rate (HR) and normalized discounted cumulative gain (NDCG). Notably, in our experiment protocol, the metrics Precision and Recall have the same tendency with Hit Rate. In the following, we report HR@10 and NDCG@10 for all the models.

\subsection{Comparing Methods}
We compare our proposed attention to state-of-the-art baselines, where we select the continuous time regularization parameter over \{1e-4, 1e-5, 1e-6\}. For the sake of fairness, the embedding size of all attentions are set to 50. In addition, we summarize the implementation details of the compared baselines as follows:
\begin{itemize}[topsep=0in,leftmargin=0em,wide=0em,parsep=0em]
    \item\textbf{GRU4REC}~\cite{HidasiICLR16} is one of state-the-of-art sequential recommendation approaches, which relies on GRU module \cite{chung2014empirical} to capture the sequential patterns in user preferences. We use the program provided by the author\footnote{https://github.com/hidasib/GRU4Rec}. Grid search of factor size over \{50, 100, ..., 300\}, learning rate over \{1e-4, 1e-3, 1e-2\} and dropout rate over \{0.1, ... ,0.5\} is performed to select the best parameters.
    
    \item \textbf{DIN}~\cite{zhou2018deep} is among the best sequential recommendation models that build upon attentions. The implementation is public and provided by the authors\footnote{https://github.com/zhougr1993/DeepInterestNetwork}. We finetune the model by grid search learning rate and dropout rate over \{0.5, 1.0, 1.5, 2.0\} and \{0.1, ... ,0.5\}.
    
    \item \textbf{SHAN}~\cite{ying2018sequential} proposes hierarchical attentions to model multi-level user representations. The source code is also provided by the authors\footnote{https://github.com/uctoronto/SHAN}. After grid search, we find the parameter settings in the original paper provide the best results in our experiments.
    
    \item \textbf{SASREC}~\cite{kang2018self} applies the Transformer~\cite{vaswani2017attention} into the sequential recommendation task and has achieved state-of-the-art performance in many benchmark datasets. We use the software provided by the authors in our experiments\footnote{https://github.com/kang205/SASRec}. We adopt the default architecture used in its original paper where three transformer blocks are used. We further search the optimal hyper-parameters by ranging learning rate over \{1e-4, 5e-4, ..., 1e-2\} and dropout rate from 0.1 to 0.5.
    
    \item \textbf{SASREC+}~\cite{xu2019self} equips SASREC with time embedding\footnote{https://github.com/StatsDLMathsRecomSys/Self-attention-with-Functional-Time-Representation-Learning}, analogous to positional embedding. This approach models continuous-time dynamics in the form of random Fourier features. In our experiments, we concatenate both time and positional embedding with the attention input embedding for the improved model accuracy. We emphasize the comparison to SASREC+, since it is probably most closely related to our work.
\end{itemize}

\subsection{Ablation Study}
We evaluate each of the two components of self-modulating attention in continuous time space, namely continuous time regularization (CTReg) and self-modulating layer (SMLayer). Ideally, our attention can be used to improve any attention mechanisms by simply stacking SMLayer and adjusting loss function with CTReg. To demonstrate the adaptability of our proposed attention, we equip DIN and SASREC progressively with SMLayer and CTReg.

Table~\ref{tab:cmp_ac} reports the results of HR@10 and NDCG@10 on the Amazon, Koubei and Tmall datasets. We can see that SMLayer achieves the consistent improvement in both HR@10 and NDCG@10 over both DIN and SASREC across all the three datasets. Especially, SMLayer helps to improve HR@10 and NDCG@10 of DIN by a margin greater than 0.01. On the other hand, CTReg can also boost the performance in most cases. One exception is on the Amazon dataset. The reason is perhaps due to the fact that Amazon is made of user reviews, and there are nearly 5.39\% users wrote the reviews of all the purchased items at the same time. Meanwhile, with using CTReg, the conditional intensity has to be enforced to fit such nuisance patterns. This will inevitably make the optimization procedure prone to overfitting, leading to the degradation of the performance.

The above experiments indicate that both the continuous time self-modulating layer and regularization are flexible and have the ability to improve the model performance. In our analysis, the reasons why our proposed attentions can further improve the recommendation accuracy is mainly due to (1) continuous time self-modulating layer models the impact of both sequential positions and continuous timestamps, so that the model expressiveness is improved; and (2) continuous time regularization enforces the intensity function dependent on the sequence timestamps, so that the intensities can accurately modulate historical behaviors to make attentive contributions to the predictions of the model.

\begin{figure}[tb!]
\begin{center}
    \centerline{\includegraphics[width=.5\textwidth]{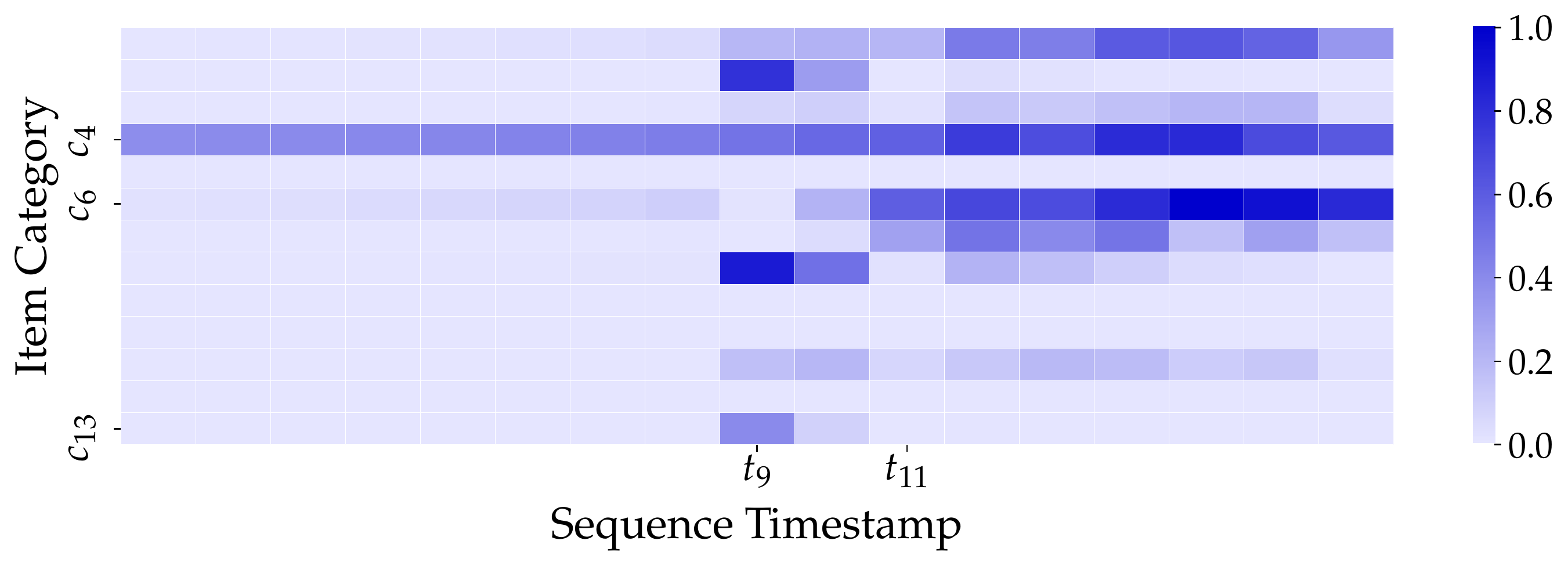}}
    \caption{\label{fig:qual_intens}Example of user preference intensities on  Koubei. The user usually purchases the items from category $c_4$ for all the time, and he/she starts to repeatedly buy the items from  category $c_6$ after time $t_9$. The darker color corresponds to the higher intensity.}
\end{center}
\end{figure}

\subsection{Quantitative Study}
This study evaluates the accuracy of the proposed method by comparing it with five state-of-the-art methods summarized previously, i.e., SHAN~\cite{ying2018sequential}, DIN~\cite{zhou2018deep}, GRU4REC~\cite{HidasiICLR16}, SASREC~\cite{kang2018self} and SASREC+~\cite{xu2019self}. Each of the methods is configured by using default parameters provided in the original paper or optimal parameters which produce the best results in grid search. We note again that for the sake of fairness, the embedding size $d$ of query, key, value matrix is set to $50$ for all attentions.

Table~\ref{tab:cmp_ac} presents HR@10 and NDCG@10 of all models on the Amazon, Koubei and Tmall datasets. This study shows that the SASREC method with our self-modulating attention outperforms on all the datasets. Note that DIN and SHAN perform worse than the results reported in their origin paper. The difference is caused by the experiment protocol. We evaluate the models on the users that are not present at training time. This setup is more difficult than that where test users and training users are overlapped.

More importantly, we emphasize the comparison between our method and SASREC+ which employs functional time embedding~\cite{xu2019self} to capture temporal dynamics. The empirical results demonstrate the superiority of our self-modulating attention. We believe this is owed to the design of intensity function, powered by neural networks, which is expressive enough to model complex evolving dynamics in the continuous time space.

\subsection{Qualitative Study}
We also qualitatively analyze the self-modulating attention outputs. Fig. \ref{fig:qual_intens} shows the intensities of a user on Koubei evolve alone the time. We note that the reported user usually purchases the items from the category $c_4$ for all the time, start to repeatedly buy the items from category $c_6$ after $t_9$.

At the first glimpse, the intensities are relatively sparse and most of them are less than one. We recall that intensities less than one will attenuate the corresponding past behavior making contributions to the output ranking score, whereas intensities equal to zero can help to derive sparse interactions between the past behaviors and future actions. As demonstrated in Theorem \ref{thm:bound}, this makes potentials for the improved generalization capacity. 

One can see that for category $c_4$, intensity function assigns larger weights to more recent behaviors. This coincides with a common experience that more recent behaviors have more impact on the future actions. On the other hand, the intensities for category $c_{13}$ remain zero for a long time, then increase to nearly 0.4 during $(t_9, t_{11})$, and finally diminish back to zero after time $t_{11}$. This demonstrates the our self-modulating attention takes the timestamps into account and has the ability to cancel the impact of the items that the user might not be interested in during certain time interval.

\section{Conclusion}
We have proposed a novel attention mechanism in continuous time space, entitled self-modulating attention, to model the non-linearly evolving dynamics, especially for user preferences. 
The key contribution lies in the generalized formulation of conventional attentions, and the development of self-modulating layer and continuous time regularization. We also provide theoretical generalization bound showing the advantage of our method.
Extensive experiments on the top-N recommendation task have demonstrated the effectiveness and flexibility of our proposed self-modulating attention mechanism in continuous time space.

\section*{Acknowledgements}
This research was supported by Shanghai Municipal Science and Technology Major Project (2021SHZDZX0102), NSFC (U19B2035, 61972250, 72061127003), and Meituan.

\bibliography{main}
\bibliographystyle{icml2021}

\appendix

\section{Proofs of Theoretical Results}
\begin{theorem}[Hoeffding's Inequality]\label{thm:hoeffding}
If $X$ is a random variable with $\mathbb{E} [X]=0$ and $X \in [a, b]$, then it holds that for any real $s>0$,
\begin{flalign}
    \mathbb{E}[e^{sX}] \le e^{s^2(b-a)^2/8}.
\end{flalign}
\end{theorem}

\begin{lemma}[Empirical Rademacher Complexity]\label{thm:radmchr}
Given a function class $\mathcal{F} = \{\mathbf{PVB}\in\mathbb{R}^{m\times{n}} \mid \mathbf{P}\in\mathbb{R}^{m\times{n}}, \mathbf{V}\in\mathbb{R}^{n\times{d}},
\mathbf{B}\in\mathbb{R}^{d\times{n}}\}$ and a set of examples $\Omega$ drawn i.i.d from the distribution $p(D)$, let $\rho|\Omega|=\sup_{\mathbf{P}\in\mathcal{F}} \sum_{(u,i)\in\Omega}\parallel\mathbf{P}_{i, *}\parallel_0$ and $\mu=\sup_{u,i,k} \mid\mathbf{P}_{u,k} (\mathbf{VB})_{k,j}\mid$, then the empirical Rademacher complexity satisfies
\begin{align}
    \hat{\mathcal{R}}_{|\Omega|}(\mathcal{F})
    \le \mu\sqrt{\frac{2\rho\ln|\mathcal{F}|}{|\Omega|}}.
\end{align}
\end{lemma}
\begin{proof}
Let $\{\sigma_{u,i}\}$ denote the Rademacher variables which are independent random variables uniformly chosen from $\{-1, 1\}$, and $\{\delta_{u,k}\}$ denote the indicator variables of which each equals to one if $\mathbf{P}_{u,k}$ is non-zero and equals to zero otherwise. 

To set up the use of Hoeffding's Inequality as in Theorem \ref{thm:hoeffding}, we start by taking the exponential of the empirical Rademacher complexity multiplied by some positive and real constant $s$:
\begin{flalign}
&\exp\Bigg( 
    s \mathbb{E}_\sigma 
        \Bigg[ 
            \sup_{(\mathbf{P},\mathbf{V},\mathbf{B})\in\mathcal{F}} 
                \sum_{(u,i)\in\Omega} 
                \sigma_{u,i} \mathbf{P}_{u,*} (\mathbf{VB})_{*, i}
        \Bigg] 
    \Bigg) \nonumber\\
&\stackrel{(a)}{\le} \mathbb{E}_\sigma\Bigg[ 
    \exp\Bigg( 
            \sup_{(\mathbf{P},\mathbf{V},\mathbf{B})\in\mathcal{F}} 
                \sum_{(u,i)\in\Omega} 
                s\sigma_{u,i} \mathbf{P}_{u,*} (\mathbf{VB})_{*, i}
        \Bigg)
    \Bigg] \nonumber\\
&\stackrel{(b)}{=} \mathbb{E}_\sigma\Bigg[ 
    \sup_{(\mathbf{P},\mathbf{V},\mathbf{B})\in\mathcal{F}} 
    \exp\Bigg( 
                \sum_{(u,i)\in\Omega} 
                s\sigma_{u,i} \mathbf{P}_{u,*} (\mathbf{VB})_{*, i}
        \Bigg)
    \Bigg] \nonumber\\
&\stackrel{(c)}{\le} \sum_{(\mathbf{P},\mathbf{V},\mathbf{B})\in\mathcal{F}} 
    \mathbb{E}_\sigma\Bigg[ 
    \exp\Bigg( 
                \sum_{(u,i)\in\Omega} 
                s\sigma_{u,i} \mathbf{P}_{u,*} (\mathbf{VB})_{*, i}
        \Bigg)
    \Bigg] \nonumber \\
&=  \sum_{(\mathbf{P},\mathbf{V},\mathbf{B})\in\mathcal{F}} 
    \mathbb{E}_\sigma\Bigg[ 
    \exp\Bigg( 
                \sum_{(u,i)\in\Omega} 
                s\sigma_{u,i} 
                \sum_k \delta_{u,k}\mathbf{P}_{u,k} (\mathbf{VB})_{k, i}
        \Bigg)
    \Bigg] \nonumber \\
&= \sum_{(\mathbf{P},\mathbf{V},\mathbf{B})\in\mathcal{F}} 
    \mathbb{E}_\sigma\Bigg[ 
    \exp\Bigg( 
                \sum_{u,i,k}
                s\sigma_{u,i} \delta_{u,k}
                 \mathbf{P}_{u,k} (\mathbf{VB})_{k, i}
        \Bigg)
    \Bigg] \nonumber \\
&\stackrel{(d)}{\le} \sum_{(\mathbf{P},\mathbf{V},\mathbf{B})\in\mathcal{F}} 
    \prod_{u,i,k} 
    \mathbb{E}_\sigma\Bigg[ 
    \exp\Bigg( 
                s\sigma_{u,i} \delta_{u,k}
                 \mathbf{P}_{u,k} (\mathbf{VB})_{k, i}
        \Bigg)
    \Bigg]. \nonumber
\end{flalign}where (a) holds due to the Jensen's Inequality; (b) holds due to monotonically increasing exponential function; (c) holds due to the nonnegativity of exponentials; and (d) holds due to the assumption of (linear) independence.

Next we apply Hoeffding's Inequality
\begin{flalign}
&\exp\Bigg( 
    s \mathbb{E}_\sigma 
        \Bigg[ 
            \sup_{(\mathbf{P},\mathbf{V},\mathbf{B})\in\mathcal{F}} 
                \sum_{(u,i)\in\Omega} 
                \sigma_{u,i} \mathbf{P}_{u,*} (\mathbf{VB})_{*, i}
        \Bigg] 
    \Bigg) \nonumber\\
&\stackrel{(a)}{\le} \sum_{(\mathbf{P},\mathbf{V},\mathbf{B})\in\mathcal{F}} 
    \prod_{u,i,d} 
    \exp\Bigg( 
                \delta_{u,k}s^2(2\mu)^2 / 8
        \Bigg) \nonumber\\
&\stackrel{(b)}{\le} \sum_{(\mathbf{P},\mathbf{V},\mathbf{B})\in\mathcal{F}}             \exp\Bigg( 
        \rho|\Omega|(s\mu)^2/ 2
    \Bigg) \nonumber\\
&=   |\mathcal{F}| \exp\Bigg( 
        {s}^2\mu^2\rho|\Omega|/ 2
    \Bigg) \nonumber
\end{flalign}where (a) holds due to $\mathbb{E} [\sigma_{u,i}\mathbf{P}_{u,k} (\mathbf{VB})_{k, i}] = 0$ and $\sigma_{u,i}\mathbf{P}_{u,k} (\mathbf{VB})_{k, i} \in [-\mu. \mu]$; and (b) holds due to that the number of nonzero entries is upper bounded by $\rho|\Omega|$.

Taking the log of both sides then dividing both sides by $s$, we have
\begin{flalign}
& \mathbb{E}_\sigma 
        \Bigg[ 
            \sup_{(\mathbf{P},\mathbf{V},\mathbf{B})\in\mathcal{F}} 
                \sum_{(u,i)\in\Omega} 
                \sigma_{u,i} \mathbf{P}_{u,*} (\mathbf{VB})_{*, i}
        \Bigg] \nonumber\\
&\le \frac{\ln|\mathcal{F}|}{s} + \frac{\mu^2\rho|\Omega|}{2} \nonumber
\end{flalign}

Note that the right hand is a function of $s$, which can be minimized by finding its derivative and setting it equal to zero. We figure out the optimum $s^\ast = \sqrt{2\ln|\mathcal{F}|/(\mu^2\rho|\Omega|)}$. After substituting this quantity back to the previous bound, we have
\begin{flalign}
&\mathbb{E}_\sigma 
        \Bigg[ 
            \sup_{(\mathbf{P},\mathbf{V},\mathbf{B})\in\mathcal{F}} 
                \sum_{(u,i)\in\Omega} 
                \sigma_{u,i} \mathbf{P}_{u,*} (\mathbf{VB})_{*, i}
        \Bigg] \nonumber\\
&\le \mu\sqrt{2\rho|\Omega|\ln|\mathcal{F}|} \nonumber
\end{flalign}

Dividing both sides by $|\Omega|$ yields the result stated in the lemma.
\end{proof}

\begin{theorem}[Generalization Bound~\cite{mohri2012foundations}]\label{thm:mohri}Given a random sample $\Omega$ drawn from $\mathbf{R}\in\mathbb{R}^{m\times{n}}$ and the loss function $\ell$ is $L$-Lipschitz, it holds that for any $\hat{\mathbf{R}}\in\mathcal{F}$
\begin{align}
    \mathbb{E} \Big[
        \ell(\mathbf{R}, \hat{\mathbf{R}}) \Big] \le \mathbb{E}_\Omega &\Big[
        \ell(\mathbf{R}, \hat{\mathbf{R}}) \Big] 
        + 2L\cdot\hat{\mathcal{R}}_{|\Omega|}(\mathcal{F}) \nonumber\\
        &+ \mathcal{O}(\sqrt{\frac{\ln(1/\delta)}{|\Omega|}})
\end{align}with probability of at least $1-\delta$.
\end{theorem}

\begin{theorem}[Generalization Bound]\label{thm:bound}Suppose that the loss function $\ell$ is $L$-Lipschitz, and for the estimate $\hat{\mathbf{R}}$ on an random example set $\Omega$ we bound $\rho|\Omega|=\sup_{\mathbf{P}\in\mathcal{F}} \sum_{(u,i)\in\Omega}\parallel\mathbf{P}_{i, *}\parallel_0$ and $\mu=\sup_{(u,i,k)\in\Omega} |\mathbf{P}_{u,k} (\mathbf{VB})_{k,j}|$, then with probability at least $1 - \delta$
\begin{align}
    \mathbb{E} \Big[
        \ell(\mathbf{R}, \hat{\mathbf{R}}) \Big]
    \le \mathbb{E}_\Omega &\Big[
        \ell(\mathbf{R}, \hat{\mathbf{R}}) \Big] 
    + \mathcal{O} \Big(
        L\mu\sqrt{\frac{C\rho\ln|\Omega|}{|\Omega|}} \nonumber\\
    &+ \sqrt{\frac{\ln(1/\delta)}{|\Omega|}}
    \Big)
\end{align}where $C = d(m+n)\log(48emn)$.
\end{theorem}
\begin{proof}
As the fact $m\gg{d}$, $\hat{\mathbf{R}}$ is actually a rank-$d$ approximation of $\mathbf{R}$, of which the VC-dimension is upper bounded by $d(M+N)\log(48eMN)$ \cite{srebro2004generalization}. Then using Sauer's Lemma, we bound $\ln|\mathcal{F}|<d(m+n)\log(48emn)\ln|\Omega|$. After combining Lemma \ref{thm:radmchr} and Theorem \ref{thm:mohri}, we can conclude the result in the theorem.
\end{proof}

\end{document}